\documentclass[letterpaper, 10 pt, conference]{ieeeconf}  

\IEEEoverridecommandlockouts                              
\usepackage{times}
\usepackage[utf8]{inputenc}
\usepackage{xcolor}
\usepackage{amsmath}

\usepackage{amsthm}
\usepackage{amsfonts}
\usepackage{amssymb}
\usepackage{graphicx}
\usepackage{float}
\usepackage[justification=centering]{caption}
\usepackage{tikz}
\usepackage[noadjust]{cite}
\usepackage{hyperref}
\usepackage{comment}
\usepackage{tabularx}

\newcommand{\R}{\mathbb{R}}
\newcommand{\Rnneg}{\R_{\geq 0}}
\newcommand{\Rpos}{\R_{>0}}

\newcommand{\Z}{\mathbb{Z}}

\DeclareMathAlphabet{\mymathbb}{U}{BOONDOX-ds}{m}{n}

\DeclareMathOperator*{\argmin}{argmin}

\newcommand{\W}{\mathcal{W}}
\newcommand{\K}{\mathcal{K}}
\newcommand{\M}{\mathcal{M}}
\newcommand{\bbW}{\mathbb{W}}
\newcommand{\cJ}{\mathcal{J}}
\newcommand{\D}{\mathfrak{D}}

\newcommand{\lp}{\left(}
\newcommand{\rp}{\right)}
\newcommand{\of}{\circ}
\newcommand{\red}{\color{red}}
\newcommand{\half}{\frac{1}{2}}

\theoremstyle{plain}
\newtheorem{thm}{Theorem}
\newtheorem{lem}[thm]{Lemma}
\newtheorem{cor}[thm]{Corollary}
\newtheorem{prop}[thm]{Proposition}
\theoremstyle{definition}
\newtheorem{defn}[thm]{Definition}

\theoremstyle{remark}

\newtheorem*{sktch}{Sketch of Proof}
\newenvironment{sketch}{\begin{sktch}}{\hfill $\qed$ \end{sktch}}

\newtheorem{problem}[thm]{\bf Problem}

\newcommand\copyrighttext{%
\footnotesize\textcopyright 2023 IEEE. Personal use of this material is permitted.  Permission from IEEE must be obtained for all other uses, in any current or future media, including reprinting/republishing this material for advertising or promotional purposes, creating new collective works, for resale or redistribution to servers or lists, or reuse of any copyrighted component of this work in other works.}
\newcommand\copyrightnotice{
\begin{tikzpicture}[remember picture,overlay]
\node[anchor=south,yshift=10pt] at (current page.south) {\fbox{\parbox{\dimexpr\textwidth-\fboxsep-\fboxrule\relax}{\copyrighttext}}};
\end{tikzpicture}}

\title{\LARGE \bf
Continuum Swarm Tracking Control: A Geometric Perspective in Wasserstein Space
}

\author{Max Emerick and Bassam Bamieh
\thanks{M. Emerick and B. Bamieh are with the Department of Mechanical Engineering,
		University of California, Santa Barbara, USA
        {\tt\small memerick@ucsb.edu, bamieh@ucsb.edu}}%
}

\begin{document}

\maketitle
\thispagestyle{empty}
\pagestyle{empty}

\begin{abstract}
We consider a setting in which one swarm of agents is to service or track a second swarm, and formulate an optimal control problem which trades off between the competing objectives of servicing and motion costs. We consider the continuum limit where large-scale swarms are modeled in terms of their time-varying densities, and where the Wasserstein distance between two densities captures the servicing cost. We show how this non-linear infinite-dimensional optimal control problem is intimately related to the geometry of Wasserstein space, and provide new results in the case of absolutely continuous densities and constant-in-time references. Specifically, we show that optimal swarm trajectories follow Wasserstein geodesics, while the optimal control tradeoff determines the time-schedule of travel along these geodesics. We briefly describe how this solution provides a basis for a model-predictive control scheme for tracking time-varying and real-time reference trajectories as well.
\end{abstract}

\copyrightnotice
\vspace{-1em}

\section{Introduction}

Low-cost sensing, processing, and communication hardware is driving the use of autonomous robotic swarms in diverse settings, including emergency response, transportation, logistics, data collection, and defense.
Large swarms in particular can have significant advantages in efficiency and robustness, but modeling these large swarms can become an issue.
For sufficiently large swarms, modeling the swarm as a density distribution over the domain (i.e. as a continuum) provides a significant model reduction as well as improved insight into the global behavior of the swarm. Thus, the development of effective motion planning and control strategies for systems of distributions is a problem of interest.

One natural mathematical setting in which to study these continuum swarm problems is the {\em Wasserstein space} of optimal transport theory. This space equips the set of normalized distributions over a domain with the {\em Wasserstein distance} -- a metric based on the cost of transport whose utility lies largely in the fact that it respects the topology of the underlying physical space\footnote{Compare, for example, the $L^p$ distance, which can be arbitrarily sensitive to spatial perturbations.}. In recent years, numerous approaches have been taken to swarm control using the Wasserstein distance and other tools from optimal transport theory. These works fall into two main approaches: one based on distributed optimization, and one based on optimal control. This work falls into the latter category.

In the distributed optimization approach, each agent communicates and acts locally to steer the entire swarm towards a given target distribution. The most relevant works here include \cite{Bandyopadhyay2014,Krishnan2018a,Inoue2021}. This approach has the advantage of being decentralized, and thus balancing the computation and communication loads in a way that is desirable in a practical implementation. However, these approaches have also been limited in their treatment of the objective and constraints: they seek to converge to the target distribution while minimizing the transport distance. In a real-world setting, it is desirable to have a model that can accommodate more general behaviors, especially since convergence to a target distribution is not always possible.

On the other hand, in the optimal control approach, a central planner controls the entire swarm to minimize a given cost function. This approach has the advantages of being able to accommodate more general objectives and constraints and yielding greater insight into the nature of optimal swarm behavior generally. While the solutions here are centralized and thus impractical for a large-scale implementation, one expects that the results from this approach will lend tools to design better distributed swarm control algorithms.

There is a large recent literature on optimal control in Wasserstein spaces, with the most relevant works in the area of multi-agent systems including \cite{Fornasier2014,Carrillo2014,Bonnet2019,Bonnet2019a, Bonnet2021,Bongini2017,Gangbo2008,Jimenez2020,Burger2021}. These works focus primarily on the analytic aspects of the problem such as existence and uniqueness of solutions and necessary conditions for optimality. They also focus on systems governed by various types of nonlocal PDEs which arise in connection with mean-field models for self-organizing swarms.

In this paper, we continue the investigation of a model for continuum swarm tracking control originally proposed in \cite{Emerick2022}. The approach we take is in line with the second group of approaches in that it is based on optimal control in Wasserstein space. However, it differs from these approaches in that it is primarily concerned with tracking as opposed to self-organization and that it uses more explicit models to obtain sharper results. In short, the focus of this work is to answer the question ``what principles underlie optimal motion planning and control in swarm-based tracking scenarios?'' In contrast to our previous work on this problem \cite{Emerick2022,Emerick2022a} which focused on the special case of swarms in one spatial dimension, this work is more general in that it treats swarms in $n$ dimensions (of course, $n$ = 1, 2, or 3 in a practical setting). In addition, this work takes a distinctively geometric approach, which not only provides more powerful tools for solving these problems but yields additional insight into the problem's underlying structure.
The main contributions of this paper are thus:
\begin{itemize}
	\item The introduction of several geometric tools for solving swarm tracking control problems,
	
	\item Analytic solutions to our model in the $n$-dimensional case where the swarm distribution is absolutely continuous and the reference distribution\footnote{In this paper we use the terms {\em density} and {\em distribution} synonymously.}  is static.
	%
\end{itemize}

We first introduce our problem formulation and present our main result. The rest of the paper is then devoted to developing the tools necessary to prove it.


\section{Problem Formulation and Main Result}


In previous work \cite{Emerick2022} we proposed a model for swarm tracking control. Since this proposed model forms the basis of this work, we briefly review the problem formulation here. For a full development and motivation, see \cite{Emerick2022}.

The key feature of our proposed problem setting is that there are two distributions, which we refer to as {\em demand} and {\em resource} respectively. The demand can be mobile or stationary and represents the spatial distribution of entities which require attention (e.g. locations to monitor, facilities to supply, independent mobile robots to be serviced, etc.). The resource represents the spatial distribution of controlled mobile agents which are able to service the needs of the demand entities. Physical space is taken to be $\R^n$ where $n$ is a positive integer. Both the demand $D$ and resource $R$ are modeled as time-varying densities
over the domain, taking the form\footnote{We conceptualize the demand and resource as parameterized curves in the space of distributions, thus we write $D$ to refer to the entire curve and $D_t$ to refer to the particular distribution at the parameter value $t$.}
\begin{equation}
	D_t(x) ~=~ d(x,t) ~+~ \sum_{i=1}^N m_i(t) \, \delta \big( x-p_i(t) \big)
\end{equation}
and similarly for $R$. Here, $x \in \R^n$ is the spatial location while $t \in [0,T]$ is the time. The function $d: \R^n \times [0,T] \to \Rnneg$ is the continuum portion of the density, which describes a continuum approximation of a large swarm. The summation represents the discrete portion of the density, which describes $N$ discrete components having mass $m_i$ at location $p_i$. In this work, we present results exclusively for continuum swarms, and so from here on we assume that $N = 0$. We also assume that both resource and demand distributions have been scaled to integrate to 1
\begin{equation}
	\int_{\R^n} R_t \, dx ~=~ \int_{\R^n} D_t \, dx ~=~ 1 .
\end{equation}
We then say that the resource and demand distributions are {\em normalized}, and denote the set of normalized distributions over $\R^n$ as $\D(\R^n)$. We also make the additional technical assumption that the resource and demand distributions are supported within a compact subset $\Omega \subset \R^n$ for all time.

The resource agents service the needs of the demand entities through an assignment process where each resource particle is coupled to a set of demand particles and vice versa. This coupling is called an {\em assignment kernel} and is one of the decision variables in our optimal control problem. The assignment kernel incurs an {\em assignment cost} which is related to the efficiency of the assignment. In the problem setting discussed in \cite{Emerick2022}, the resource provides communication services to the demand, and the assignment cost is thus proportional to the squared distance between coupled resource and demand particles. Supposing that the communication loads must be balanced, each demand particle must be serviced, and the coupling is chosen optimally, we obtain the assignment cost
\begin{equation}
	\begin{split}
		&\inf_\K \, \int_{\R^n \times \R^n} \Vert x - y \Vert_2^2 \, \K(x,y,t) \, dx \, dy \\
		& \quad \text{s.t.} \quad \textstyle{\int \K \, dy = R} \\
		& \quad \phantom{\text{s.t.}} \quad \textstyle{\int \K \, dx = D} .
	\end{split}
\end{equation}
It turns out that this assignment cost is exactly the square of the 2-Wasserstein distance $\W_2$ of optimal transport theory \cite{Santambrogio2015}. Endowed with $\W_2$, the space of normalized distributions $\mathfrak{D}(\R^n)$ becomes a metric space, which we call {\em 2-Wasserstein space} and denote $\bbW_2$. This will be the setting that we work in in this paper.

It is important to note that while the Wasserstein distance is often motivated as a ``transport cost'' in optimal transport theory, that is not the way it appears in this problem setting. Here, the Wasserstein distance represents the cost of servicing, which is proportional to the distance between coupled particles, but is {\em unrelated to motion}. There is indeed a physical cost of motion in our problem setting as well, but this is quantified differently, as we describe next.

In order to track the demand, the resource swarm is controlled through a time-varying vector field $V$. The equations of motion of the system are given by the {\em transport equation}
\begin{equation} \label{transport_eqn}
	\partial_t R_t(x) ~=~ - \nabla \cdot \big( V_t(x) \, R_t(x) \big) .
\end{equation}
The vector field $V$ is the main decision variable in this optimal control problem. $V$ can be used to steer the resource distribution closer to the demand distribution, lowering the assignment cost. However, this incurs a {\em motion cost} which is related to the total energy expended in motion. In our problem setting, the motion cost represents the total aerodynamic drag on the resource swarm, taking the form
\begin{equation}
	\int_{\R^n} \big| V_t(x) \big|_2^2 \, R_t(x) \, dx ~=:~ \big\Vert V_t \big\Vert_{L^2(R_t)}^2 .
\end{equation}

Lastly, we have an {\em objective function} which defines optimal behavior for the resource swarm. The objective function is given by the integral of the assignment cost plus the motion cost over the time horizon $[0,T]$. We remark again that the assignment cost and the motion cost are competing objectives: if $R_t$ and $D_t$ are far apart, then the assignment cost will be large, but this cost can be reduced at the expense of motion (and vice versa). To capture the tradeoff between assignment and motion costs, we use a weighting parameter $\alpha$ to control the relative importance of the two costs. All in all, the proposed problem is to find the control which minimizes the total cost of a maneuver. The problem is stated formally as follows.
\begin{problem}[\em Original Formulation] 															\label{orig}
	Given an initial resource distribution $R_0$ and demand trajectory $D$ over $[0,T]$, solve
	\begin{equation}
		\begin{aligned} 
			\inf_{V} \int_0^T & \W_2^2 \big( R_t,D_t \big) ~+~ \alpha \, \big\Vert V_t \big\Vert_{L^2(R_t)}^2 \, dt  \\
			\mbox{s.t.} \quad
			&\partial_t R_t(x) = -\nabla \cdot \big( V_t(x) \, R_t(x) \big) . \\
		\end{aligned}
		\label{orig_eq}
	\end{equation}
\end{problem}

As stated, this problem is a nonlinear, infinite-dimensional optimal tracking problem: $R$ plays the role of the state, $D$ the reference trajectory, and $V$ the control input. The assignment cost (given by $\W_2^2$) plays the role of the tracking error, while the motion cost (given by $\| V_t \|^2$) plays the role of the control energy. The transport equation is the dynamic constraint.

Notice that the decision variable $\K$ does not appear explicitly in this problem formulation. This is because optimal solutions always have $\K$ chosen optimally at each (static) instant in time. Assuming each of these static subproblems to be solved, we obtain the above formulation. For further discussion on this, see \cite{Emerick2022}.

A cartoon depiction of this problem is shown below in Figure \ref{model.fig}.

\begin{figure}[!ht]
	\centering
	\includegraphics[width=\linewidth]{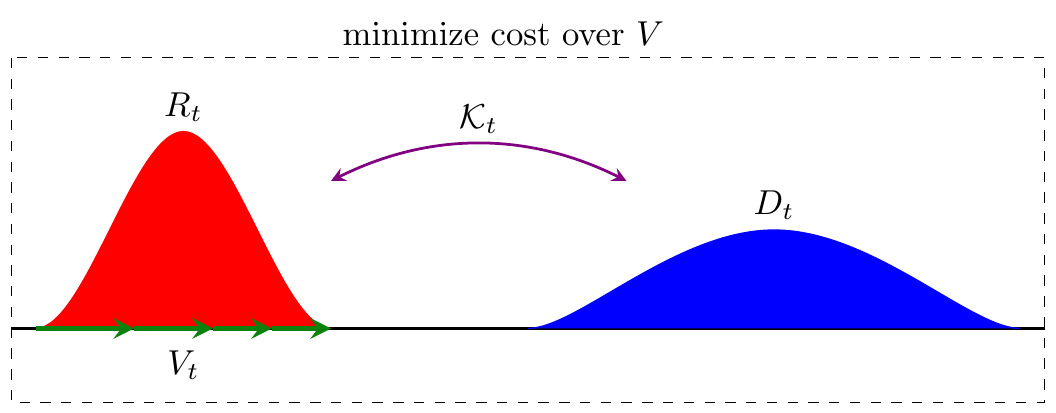}
	\caption{The resource $R$ is paired to the demand $D$ by the assignment kernel $\K$ and is transported by the vector field $V$. The cost is minimized over the whole maneuver.}
	\label{model.fig}
\end{figure}

\subsection{Main Result}

The main result of this paper is a complete characterization of solutions to Problem \eqref{orig} when $R_0$ is absolutely continuous and $D$ is constant in time. To state it, we first need to introduce some background.

\begin{defn}[Pushforward]
	Let $\mu$ be a density defined on $\Omega_1$ and $f$ be a function from $\Omega_1$ to $\Omega_2$. The {\em pushforward of $\mu$ through $f$} is a density $f_\# \mu$ defined on $\Omega_2$ such that
	\begin{equation} \label{pushforward_eq}
		\int_{\Omega_2} \psi(y) \, \big( f_\# \mu \big) (y) \, dy ~=~ \int_{\Omega_1} \big( \psi \of f \big) (x) \, \mu(x) \, dx
	\end{equation}
	for all test functions $\psi: \Omega_2 \to \R$.
\end{defn}

The pushforward can be conceptualized as the density formed by ``moving the mass in $\mu$ forward through $f$''. The above definition in terms of integration against test functions is equivalent to other standard definitions, but will be more useful to us.

The notion of the pushforward allows us to define {\em transport maps} between densities.

\begin{defn} [Transport Maps]
	Given two normalized densities $\mu$, $\rho$ on $\R^n$, a {\em transport map from $\mu$ to $\rho$} is a map $M: \R^n \to \R^n$ such that $\rho = M_\# \mu$. $M$ is said to be {\em optimal} if it minimizes (over all such maps) the functional
	\begin{equation}
		\int_{\R^n} \big| M(x) - x \big|^2 \, \mu(x) \, dx ~=:~ \big\| M - \mathcal{I} \big\|_{L^2(\mu)}^2 .
	\end{equation}
\end{defn}

Lastly, we introduce the concept of a geodesic.

\begin{defn}[Wasserstein Geodesic]
	Given two normalized densities $\mu$, $\rho$ on $\R^n$, a (constant-speed) {\em Wasserstein geodesic from $\mu$ to $\rho$} is a curve $\gamma: [0,1] \to \bbW_2$ with $\gamma_0 = \mu$ and $\gamma_1 = \rho$ having the property
	\begin{equation}
		\W_2(\gamma_{t_1},\gamma_{t_2}) ~=~ | t_1 - t_2 | \, \W_2(\mu,\rho)
	\end{equation}
	for all $t_1,t_2 \in [0,1]$. We write $\Gamma(\mu,\rho)$ to denote the range of the geodesic (i.e. the set of points in $\bbW_2$).
\end{defn}

Intuitively, a geodesic is a curve which achieves the minimum length between $\mu$ and $\rho$. That these curves always exist is a useful but nontrivial fact \cite{Santambrogio2015}.
	
We can now state our main result.

\begin{thm} \label{main_thm}
	When $R_0$ is absolutely continuous and $D$ is constant in time, the solution to Problem~\ref{orig} is given implicitly by the feedback controller
	\begin{equation} \label{opt_controller}
		V_t ~=~ -f(t) \, \big( \mathcal{I} - \bar{M}_t \big) / \alpha ,
	\end{equation}
	where $\mathcal{I}$ is the identity map on $\R^n$, $\bar{M}_t$ is the optimal transport map taking $R_t$ to $D$, and
	\begin{equation}
		f(t) ~=~ \sqrt{\alpha} \tanh \left( (T-t) / \sqrt{\alpha} \right) .
	\end{equation}
	The trajectory generated by this control law takes the form
	\begin{equation}
		R_t ~=~ \left[ (1-\sigma(t)) \, \mathcal{I} ~+~ \sigma(t) \, \bar{M}_0 \right]_\# R_0
	\end{equation}
	where
	\begin{equation}
		\sigma(t) ~=~ 1 - \cosh \left( (T-t) / \sqrt{\alpha} \right) .
	\end{equation}
	In particular, $R$ follows the Wasserstein geodesic from $R_0$ to $D$.
	Furthermore, this solution attains the cost
	\begin{equation} \label{opt_cost}
		\cJ(R_0,T;\alpha;D) = \frac{1}{2} \W_2^2(R_0,D) \sqrt{\alpha} \tanh \lp T / \sqrt{\alpha} \rp .
	\end{equation}
\end{thm}

In other words, the optimal control steers the state along the shortest path to the reference distribution at a rate which depends on the initial distance $\W_2(R_0,D)$, $T$, and $\alpha$.

From a practical standpoint, we point out that even though this solution is based on a static reference, it provides a basis for a receding-horizon Model-Predictive Control (MPC) scheme for tracking time-varying, stochastic, and apriori unknown demand trajectories. This MPC scheme would be similar in spirit to that described in \cite{Limon2008}, where tracking of a time-varying reference is achieved by treating it as a piecewise constant signal. In this control scheme, the demand distribution would be updated at each timestep, with the optimal control \eqref{opt_controller} being recomputed and applied each time the demand is updated. Developing a control scheme of this sort will be the subject of future work.

From a theoretical standpoint, it is somewhat surprising that the problem \eqref{orig_eq} admits tractable solutions at all, let alone the highly structured analytic solutions we see here. At first glance, problem \eqref{orig_eq} is a nonlinear, infinite-dimensional, multi-level optimization problem, and we might expect that computationally-intensive numerical solutions would be the only line of approach. However, this is not the case. The reason that this problem turns out to be nice is that it has a rich geometry, which we can leverage for some very powerful tools. It is worth giving this geometry attention then, as it seems that similar tools may be applied to solve other sorts of nonlinear optimization problems as well. Thus, the rest of this paper will focus on developing the geometric picture and the tools necessary to demonstrate our main result.


\section{The Geometric Picture}


The central idea in this section is that there are two representations for our system, which we call {\em Eulerian} and {\em Lagrangian}. The Eulerian representation describes the system in terms of particle densities, while the Lagrangian representation describes the system in terms of particle coordinates. These two representations are equivalent, and we can learn a great deal about our problem by formalizing this. In this section we describe these two representations and their correspondence following \cite{Otto2001}. In the next section we will use the tools developed here to reformulate and solve our problem.


\subsection{The Eulerian Setting}

In this setting we study the evolution of densities in the Wasserstein space $\bbW_2$.
%
%
It turns out that $\bbW_2$ forms an infinite-dimensional Riemannian manifold (with boundary) \cite{Ambrosio2005}, as we explain next. Points in $\bbW_2$ represent normalized densities, which we denote in this section by $\rho$, $\mu$, or $\chi_t$ (where $\chi: [0,T] \to \bbW_2$ is a curve in $\bbW_2$).
%
%
Recall that the derivative of a curve $\chi$ at time $t$ is formalized as the tangent vector $\chi_t'$, and the set of all tangent vectors at a given point $\rho$ forms the {\em tangent space} at $\rho$, denoted $\mathcal{T}_{\rho} \bbW_2$. Intrinsically, this tangent space is identified with the set of variations on the density $\rho$, roughly taking the form of signed distributions on $\R^n$ having zero total mass \cite{Otto2001}.
%
However, since we frequently generate curves in $\bbW_2$ via the transport equation
\begin{equation} \label{transport_eqn_2}
	\chi_t' ~=~ - \nabla \cdot (v_t \, \chi_t) ,
\end{equation}
it is natural to ask what the relation is between the intrinsic tangent vectors (i.e. the variations on $\chi_t$) and objects of the form $ - \nabla \cdot \big( v_t \, \chi_t \big)$. It turns out that every sufficiently regular curve can be generated this way (as we show next) and thus we choose to identify the tangent space $\mathcal{T}_\rho \bbW_2$ with the set of objects $\{ - \nabla \cdot ( v \, \rho) ~:~ v: \R^n \to \R^n \}$.

\begin{lem}[{\cite[Theorem~5.14]{Santambrogio2015}}] \label{tan_vel_lem}
	A curve $\chi: [0,T] \to \bbW_2$ is absolutely continuous\footnote{Note that absolute continuity is defined both for distributions and for curves. We will state which definition we are using in each context.} if and only if there exists a time-varying vector field $v$ with
	\begin{equation}
		\int_0^T \Vert v_t \Vert_{L^2(\chi_t)} \, dt ~<~ \infty
	\end{equation}
	 such that $(\chi,v)$ satisfy the transport equation \eqref{transport_eqn_2}. Furthermore, the speed of the curve $|\chi_t'|$ satisfies
	 \begin{equation}
	 	|\chi_t'| ~=~ \min_{\substack{\chi_t' = - \nabla \cdot(v_t \chi_t)}} \Vert v_t \Vert_{L^2(\chi_t)}
	 \end{equation}
	 for almost all $t$.
\end{lem}

Observe that this lemma not only relates absolutely continuous curves in $\bbW_2$ to solutions of the transport equation, but by giving us a formula for the speed, identifies the metric tensor in $\bbW_2$ as well
\begin{equation} \label{wass_metric_tensor}
	\big\langle \tau_1 , \tau_2 \big\rangle_\rho ~=~ \min_{\substack{\tau_1 = - \nabla \cdot (v_1 \rho) \\ \tau_2 = - \nabla \cdot (v_2 \rho)}} \big\langle v_1 , v_2 \big\rangle_{L^2(\rho)} .
\end{equation}
Recall that the metric tensor completely defines the geometry of the manifold, providing the notions of length, speed, distance, geodesics, and curvature.

\subsection{The Lagrangian Setting}

In this setting, we study the evolution of particle coordinates in the space of maps. Recall that the trajectory of a particle in a vector field $v$ is described by an {\em integral curve} of the vector field, and that the collection of all integral curves gives us the {\em flow map} $\phi$:
%
%
\begin{equation} \label{flow_eqn}
	\begin{split}
		\partial_t \phi_t(x) &= v_t(\phi_t(x)) \\
		\phi_0(x) &= x.
	\end{split}
\end{equation}
Recall that $\phi_t(x)$ describes the position at time $t$ of the particle that started from position $x$ at time 0. Thus, at a fixed instant in time, $\phi_t$ describes the coordinates of all particles as a map from $\R^n$ to $\R^n$. Thus the state of the system is represented as a point in the space of maps, denoted $\mathbb{M}$.

The space $\mathbb{M}$ also forms an infinite-dimensional Riemannian manifold. (In fact, $\mathbb{M}$ forms a vector space, which tells us a bit more.) Points in $\mathbb{M}$ represent maps from $\R^n$ to $\R^n$, and are denoted by $M$ (or $\phi_t$ if the map represents the flow of a particular velocity field $v$). The tangent space at $M$ is the set of variations $u$ on the map $M$. Since $\mathbb{M}$ is a vector space, this set is isomorphic to the space of maps itself, $\mathcal{T}_M \mathbb{M} \cong \mathbb{M}$. This allows us to endow this space with the inner product
\begin{equation}
	\langle u_1,u_2 \rangle_{L^p(\mu)} ~:=~ \int_{\R^n} u_1^T(x) \, u_2(x) \, \mu(x) \, dx ,
\end{equation}
where $\mu$ is some fixed reference distribution. Notice that an inner product gives a constant metric tensor, implying that $\mathbb{M}$ has zero curvature, i.e. that $\mathbb{M}$ is {\em flat}.

%

\subsection{Correspondence Between Settings}

The correspondence between the Eulerian and Lagrangian settings is identified by the following lemma.

\begin{lem}[{\cite[Theorem~4.4]{Santambrogio2015}}] \label{transport_pushforward_lem}
	Let $\chi_0 \in \bbW_2$ be an absolutely continuous initial distribution and $v: \R^n \times [0,T] \to \R^n$ be a sufficiently regular time-varying vector field.
	Then $(\chi, v)$ satisfy the transport equation \eqref{transport_eqn_2} if and only if
	\begin{equation} \label{transport_pushforward_eqn}
		\chi_t ~=~ \left[ \phi_t \right]_\# \chi_0 ,
	\end{equation}
	where $\phi$ is the flow map of $v$.
\end{lem}

There are several important points to make here. First, this correspondence is induced by the pushforward: given a fixed reference distribution $\mu$, the pushforward defines a function $\Pi$ which takes each map $M \in \mathbb{M}$ to a distribution $M_\# \mu \in \bbW_2$. Thus the pushforward induces a map between manifolds $\Pi: \mathbb{M} \to \bbW_2$.
Second, Lemma \ref{transport_pushforward_lem} is a statement about how dynamics transform under this map $\Pi$. Taking $\mu = \chi_0$ to be the reference distribution, the dynamics \eqref{transport_eqn_2} and \eqref{flow_eqn} are equivalent under $\Pi$ in the sense that the curves generated by these dynamics agree under this transformation. Thus we write
\begin{equation}
	\Pi: \phi_t \mapsto \chi_t = \left[ \phi_t \right]_\# \chi_0 .
\end{equation}
Third, knowing how the dynamics transform tells us how tangent vectors transform as well. Recall that the map between tangent spaces $\mathcal{T}_M \mathbb{M}$ and $\mathcal{T}_{\Pi(M)} \bbW_2$ is the {\em differential} of the map $\Pi$ at $M$, denoted $D \Pi (M)$.
Lemma \ref{transport_pushforward_lem} tells us immediately that $D \Pi (\phi_t): \phi_t' \mapsto \chi_t'$. We can manipulate this into an explicit form as follows. Using \eqref{flow_eqn}, we can write $v_t = \phi_t' \of \phi_t^{-1}$. Substituting this into \eqref{transport_eqn_2} along with \eqref{transport_pushforward_eqn} yields
\begin{equation}
	\chi_t' ~=~ - \nabla \cdot \Big( (\phi_t' \of \phi_t^{-1}) \big( [\phi_t]_\# \chi_0 \big) \Big) .
\end{equation}
From this, we deduce that the differential is
\begin{equation} \label{differential}
	D \Pi (M) : u \mapsto - \nabla \cdot \Big( \big(u \of M^{-1} \big) \big( M_\# \mu \big) \Big) .
\end{equation}
We take this moment to point out that we will use representations of tangent vectors in $\mathcal{T}_M \mathbb{M}$ in two different coordinate systems, related by $u = v \of M$. For example, in these transformed coordinates, the differential takes the form
\begin{equation} \label{trans_differential}
	D \Pi (M) : v \mapsto - \nabla \cdot \big( v \, ( M_\# \mu ) \big) .
\end{equation}
We will always use the letters $u$ and $v$ for tangent vectors in each of these coordinate systems.

The next natural question to ask is what properties the map $\Pi$ has (e.g. continuous, differentiable, 1-1, onto, etc.). We can see that $\Pi$ is differentiable (and thus continuous) from \eqref{trans_differential}. It turns out that $\Pi$ is not 1-1, which we can see by observing that $M_\# \mu = [M \of \mathcal{P}]_\# \mu$ for any measure-preserving transformation $\mathcal{P}$. However, $\Pi$ is onto, provided that the reference distribution $\mu$ is chosen to be absolutely continuous. We can see this because we can define a right inverse for $\Pi$ as
%
%
%
%
%
%
\begin{equation} \label{right_inv_eqn}
	\Pi^{-R} : \rho \mapsto \argmin_{\Pi(M) = \rho} \Vert M - \mathcal{I} \Vert_{L^2(\mu)} .
\end{equation}
To show that this is actually a right inverse, we need to show that this optimization problem attains a unique minimum for all $\rho$. To see this, observe that \eqref{right_inv_eqn} is actually just the Monge problem of optimal transport theory: the $L^2(\mu)$-norm of $(M - \mathcal{I})$ is the square root of the transport cost, the constraint $\Pi(M)=\rho$ expresses that $M_\# \mu = \rho$, i.e. that $M$ is a transport map from $\mu$ to $\rho$, and thus the minimizing argument $\bar{M} = \Pi^{-R}(\rho)$ is the optimal transport map. Recall that when $\mu$ is absolutely continuous, a unique minimizer to the Monge problem is guaranteed to exist \cite[Theorem~1.22]{Santambrogio2015}, and thus $\Pi^{-R}$ is indeed a right inverse. Therefore $\Pi$ and $\Pi^{-R}$ define a 1-1 correspondence between $\bbW_2$ and the set of optimal transport maps in $\mathbb{M}$, which we denote $\mathbb{O}$
\begin{equation}
	\begin{tikzpicture}[baseline={([yshift=-.5ex]current bounding box.center)}]
		\node at (-1.4,0) {$\mathbb{W}_2$};
		\node at (0.3,-0.025) {$\mathbb{O}$};
		\draw [-stealth](-1,-0.06125) -- (0,-0.06125);
		\draw [-stealth](0,0.06125) -- (-1,0.06125);
		\node at (-0.5,0.25) {${\scriptstyle \Pi}$};
		\node at (-0.5,-0.25) {${\scriptstyle \Pi^{-R}}$};
	\end{tikzpicture} .
\end{equation}
By understanding the structure of $\mathbb{O}$ we can therefore understand the structure of $\bbW_2$. We start with the following characterization.

\begin{prop}[{\cite[Theorems~1.22,~1.48]{Santambrogio2015}}]
	A transport map is optimal if and only if it is equal to the gradient of a convex function.
\end{prop}

Thus we identify $\mathbb{O} = \{ M = \nabla F : F ~ \text{convex} \}$. Remarkably, this implies that $\mathbb{O}$ is {\em independent of the distribution $\mu$}\footnote{Note, however, that the associations $\Pi$ and $\Pi^{-R}$ are {\em not} independent of $\mu$, although this is suppressed in the notation for simplicity.}.
The picture is completed by identifying the kernel of the transformation $\Pi$. Recall our earlier counterexample that $M_\# \mu = [M \of \mathcal{P}]_\# \mu$ for any measure-preserving transformation $\mathcal{P}$, thus any such $\mathcal{P}$ is in the kernel. In fact, this is a complete characterization of the kernel.

\begin{prop}[{\cite[Theorem~1.53]{Santambrogio2015}}]
	Consider a map $M: \Omega \to \R^n$ where $\Omega \subset \R^n$ is convex, and $\mu$ an absolutely continuous distribution on $\Omega$. Then there exists a convex function $F: \Omega \to \R^n$ and a map $\mathcal{P}: \Omega \to \Omega$ preserving $\mu$ such that $M = ( \nabla F) \of \mathcal{P}$. Furthermore, $\nabla F$ is unique up to $\mu$-a.e. equivalence.
\end{prop}

To summarize, then: $\Pi$ maps $\mathbb{M}$ onto $\bbW_2$. The kernel $\text{ker}(\Pi) = \{ \mathcal{P} \}$ forms a group. Two elements in $\mathbb{M}$ are equivalent under $\Pi$ if they are related by an element of $\{ \mathcal{P} \}$. Each equivalence class $[M]$ has a unique representative $\bar{M}$ in $\mathbb{O}$. Then $\mathbb{O}$ can be identified with the quotient space $\mathbb{M} / \{ \mathcal{P} \}$. Figure \ref{rie_sub} below shows a cartoon depiction of this situation.


\begin{figure}[!ht]
	\centering
	\includegraphics[width=\linewidth]{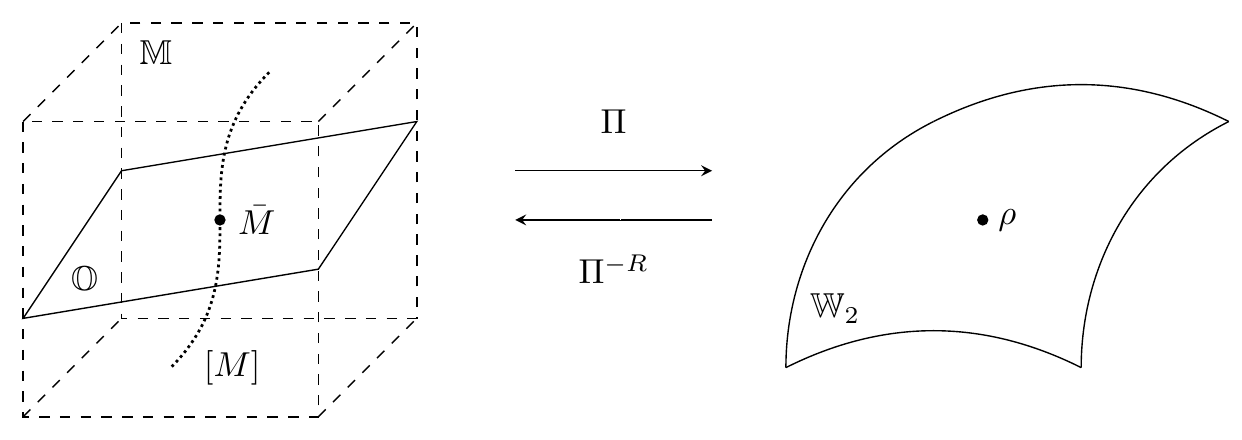}
	\caption{Relationship between $\mathbb{M}$, $\mathbb{O}$, $\bbW_2$. $\Pi$ maps $\mathbb{M}$ onto $\bbW_2$. Each equivalence class $[M]$ is the preimage of a point $\rho$ in $\bbW_2$. Each equivalence class has a unique representative $\bar{M} \in \mathbb{O}$. $\mathbb{O}$ is a quotient space of $\mathbb{M}$ and is 1-1 with $\bbW_2$.}
	\label{rie_sub}
\end{figure}

In this way, $\mathbb{O}$ is an embedding of $\bbW_2$ in $\mathbb{M}$. However, $\mathbb{O}$ carries a different geometry than that of the latent space. The natural geometry for $\mathbb{O}$ (and the one that makes it isometric to $\bbW_2$) is given by minimizing the metric tensor over all elements in the equivalence class of each point. Formally, let $\Psi = \Pi^{-R} \of \Pi$ be the projection from $\mathbb{M}$ to $\mathbb{O}$. Then the metric tensor in $\mathbb{O}$ is defined to be
\begin{equation}
	\langle w_1 , w_2 \rangle_{\bar{M}} ~=~ \min_{\substack{D \Psi (M)[u_1] = w_1 \\ D \Psi (M)[u_2] = w_2}} \langle u_1 , u_2 \rangle_{M} .
\end{equation}
We can see that this coincides with the metric tensor in $\bbW_2$ \eqref{wass_metric_tensor} by applying our coordinate change $u = v \of M$ and using properties of the pushforward
\begin{multline}
	\langle u_1 , u_2 \rangle_M ~=~ \langle u_1 , u_2 \rangle_{L^2(\mu)} ~=~ \langle v_1 \of M , v_2 \of M \rangle_{L^2(\mu)} \\
	~=~ \langle v_1 , v_2 \rangle_{L^2({\M}_\# \mu)} ~=~ \langle v_1 , v_2 \rangle_{L^2(\rho)} ,
\end{multline}
so that by taking the minimum of both sides,
\begin{multline}
	\langle w_1 , w_2 \rangle_{\bar{M}} ~:=~ \min \, \langle u_1 , u_2 \rangle_M \\ ~=~ \min \, \langle v_1 , v_2 \rangle_{L^2(\rho)} ~=:~ \langle \tau_1 , \tau_2 \rangle_\rho .
\end{multline}

In the language of differential geometry, the map $\Pi$ is called a {\em Riemannian submersion}. This means that the differential $D\Pi(M)$ is an orthogonal projection onto each tangent space $\mathcal{T}_{\Pi(M)}\bbW_2$. The restriction of $D\Pi(M)$ to the orthogonal complement of its kernel is thus an isometry between $\text{ker} (D\Pi(M))^\perp$ and $\mathcal{T}_{\Pi(M)} \bbW_2$. A tangent vector in $\text{ker}(D\Pi(M))$ is called {\em vertical}, while a tangent vector in $\text{ker} (D\Pi(M))^\perp$ is called {\em horizontal}. We have the following characterization of vertical and horizontal tangent vectors.

\begin{prop}[{\cite[Section~4.1]{Otto2001}}]
	A tangent vector $v \in \mathcal{T}_M \mathbb{M}$ is vertical with respect to $\Pi$ if and only if $\nabla \cdot (v \rho) = 0$, where $\rho = M_\# \mu$. A tangent vector $v \in \mathcal{T}_M \mathbb{M}$ is horizontal if and only if $v = \nabla H$ for some scalar function $H$.
\end{prop}


The machinery of Riemannian submersions now allows us to develop a number of useful results. First, we compute the differential of $\Pi^{-R}$. Due to the isometries between $\bbW_2$ and $\mathbb{O}$ and their tangent spaces, we know that $D\Pi^{-R}(\rho)$ must preserve the norms of tangent vectors. Since $D\Pi(M)$ is an orthogonal projection onto each tangent space, there is exactly one tangent vector in the preimage of $\tau$ with the same norm (i.e. that which is horizontal). Thus we can write
\begin{equation}
	D\Pi^{-R}(\rho): \tau \mapsto \argmin_{D\Pi(\bar{M})[v] = \tau} \| v \|_{L^2(\rho)} .
\end{equation}
The subsequent facts also follow directly from the properties of $\Pi$ as a Riemannian submersion \cite{Otto2001}:
\begin{enumerate}
	\item If $g$ is geodesic in $\mathbb{M}$ and $g_t'$ is horizontal for some t, then $g_t'$ is horizontal for all $t$. In this case, $g$ is called a {\em horizontal geodesic}.
	\item The image of a horizontal geodesic $g$ under a Riemannian submersion is a geodesic $\gamma$ in the codomain.
\end{enumerate}

Using these facts, we can now write down an expression for a geodesic in $\bbW_2$. Take two distributions $\mu$ and $\rho$ in $\bbW_2$, with $\mu$ absolutely continuous. Take $\mu$ as our fixed reference distribution. Let $\mathcal{I}$ be the identity map and $\bar{M} = \Pi^{-R}(\rho)$ be the optimal transport map taking $\mu$ to $\rho$. Since $\mathbb{M}$ is flat, the geodesic from $\mathcal{I}$ to $\bar{M}$ in $\mathbb{M}$ is given by
\begin{equation}
	g_t ~=~ (1-t) \, \mathcal{I} ~+~ t \bar{M} .
\end{equation}
Furthermore, we have $u_t := g_t' = \bar{M} - \mathcal{I}$, or $v_t = \bar{M} \of g_t^{-1} - g_t^{-1}$ in our transformed coordinates. Observe that since $g_0 = \mathcal{I}$, $v_0 = \bar{M} - \mathcal{I}$. Since $\bar{M}$ is an optimal transport map, $\bar{M} = \nabla F$ for some convex function $F$. Similarly, $\mathcal{I} = \nabla \half | \cdot |^2$, and thus we can write
\begin{equation}
	v_0 ~=~ \nabla \left( F - \half | \cdot |^2 \right) ~=~ \nabla H .
\end{equation}
Thus $v_0$ is horizontal, so $v_t$ is horizontal for all $t$, so $g$ is a horizontal geodesic in $\mathbb{M}$. Therefore the image of $g$ under $\Pi$
\begin{equation} \label{geodesic_eqn}
	\gamma_t ~=~ \left[ (1-t) \, \mathcal{I} ~+~ t \bar{M} \right]_\# \mu
\end{equation}
is a geodesic in $\bbW_2$. This completes our development of the geometric picture.

\section{Geometric Formulation of Problem}

The first step in approaching the higher-dimensional case is a geometric formulation of our original problem. In this formulation, instead of considering minimizing control inputs, we consider minimizing curves in Wasserstein space.
To do this, we need to pose the problem explicitly in $\bbW_2$. Examining Problem~\ref{orig}, we see that the assignment cost is already an object in $\bbW_2$, but we still need to transform the motion cost and constraint. By Lemma \eqref{tan_vel_lem}, we see that the motion cost is equal to the square of the speed of the curve $R$ and the constraint becomes simply that the curve $R$ be absolutely continuous (A.C.). We call this the {\em geometric formulation}.


\begin{problem}[\em Geometric Formulation]																					 \label{gmod}
	Given an initial resource distribution $R_0$ and demand trajectory $D$ over $[0,T]$, solve
	\begin{equation} \label{gmod_eq}
		\begin{split}
			&\inf_{R} \; \int_0^T \W_2^2(R_t,D_t) ~+~ \alpha \, |R_t^\prime|^2 \, dt \\
			&\quad \text{s.t.} \quad R: [0,T] \to \bbW_2 ~~ \text{A.C.}
		\end{split}
	\end{equation}
\end{problem}

Notice that this formulation has no explicit dependence on $V$. However, since $V$ is related to the tangent of $R$ via $D \Pi^{-R}$,
it is straightforward to recover one from the other. When the initial resource distribution $R_0$ is absolutely continuous, the two problems are equivalent.

\begin{lem}
	Supposing $R_0$ is absolutely continuous, Problem~\ref{orig} and Problem~\ref{gmod} are equivalent. The solutions to the two problems attain the same cost and are related by
	\begin{align}
		\bar{V}_t &~=~ \argmin_{\bar{R}_t' = - \nabla \cdot ( V_t  \bar{R}_t )} \| V_t \|_{L^2(\bar{R}_t)} \label{vbar_from_rbar} \\
		\bar{R}_t &~=~ \left[ \bar{\phi}_t \right]_\# R_0 , \phantom{\Big|} \label{rbar_from_vbar}
	\end{align}
	where $\bar{\phi}$ is the flow generated by $\bar{V}$.
\end{lem}

\begin{sketch}
	Taking $R_0$ to be the fixed reference distribution, observe that $\bar{R}_t = \Pi ( \bar{\phi}_t )$ and that $\bar{V}_t = D \Pi^{-R}(\bar{R}_t) [ \bar{R}_t' ]$. Then $\bar{R}$ and $\bar{\phi}$ are corresponding curves in $\bbW_2$ and $\mathbb{O}$ respectively, so by the isometry between these spaces, $\bar{V}$ and $\bar{R}$ attain the same cost.
	%
\end{sketch}

\vspace{0.5em}
\section{Continuous Resource / Static Demand}

In the simple case where the resource distribution is absolutely continuous and the demand distribution is static, the solution has a particularly nice form. Specifically, the solution $\bar{R}$ follows the geodesic from $R_0$ to $D$.

\begin{lem} \label{soln_follows_geo}
	Consider Problem~\ref{gmod} and suppose that $R_0$ is absolutely continuous and $D$ is constant in time. Then the solution $\bar{R}$ is such that $\bar{R}_t \in \Gamma(R_0,D)$ for all $t \in [0,T]$, where $\Gamma$ is the range of the Wasserstein geodesic.
\end{lem}

\begin{sketch}
	For any feasible solution $R_1$, we can construct a solution $R_2$ with $|R_2'| = |R_1'|$ which flows along the geodesic. We find that $\text{cost}(R_2) \leq \text{cost}(R_1)$, so if optimal solutions exist, one must flow along the geodesic.
\end{sketch}

This result should be intuitive. In balancing the motion penalty with the assignment cost (i.e. the distance to $D$), the optimal solution will minimize distance using the least effort possible, and will thus move along the shortest path between $R_0$ and $D$. Since this fully defines the path that the solution takes, we can now parameterize the problem with a single variable representing the distance along this path. This allows us to reduce the problem to a familiar form.

Let $\sigma(t)$ define the fraction of the geodesic which has been traversed by time $t$. Then the following hold:
\begin{align}
	|R_t'| &~=~ \W_2(R_0,D) \, \sigma'(t) \\
	\W_2(R_t,D) &~=~ \W_2(R_0,D) \, \big( 1 - \sigma(t) \big) .
\end{align}
Defining
\begin{align}
	\zeta &~:=~ \W_2(R_0,D) \\
	\eta(t) &~:=~ \W_2(R_0,D) \, \sigma(t) \\
	u(t) &~:=~ \W_2(R_0,D) \, \sigma'(t) ,
\end{align}
and substituting into \eqref{gmod_eq}, we obtain
\begin{align}
	& \inf_u\; \int_0^T \big( \zeta(t) - \eta(t) \big)^2 + \alpha u^2(t) \,dt \\
	& \quad \text{s.t.} \quad \dot{\eta}(t) = u(t) .
\end{align}
Observe that this is exactly the form of the scalar linear-quadratic tracking problem discussed in Appendix \ref{slqt_prob}. Then by Corollary \ref{scalar_lq_static_soln} and equation \eqref{geodesic_eqn} we obtain the following solution to Problem~\ref{gmod}.

\begin{lem}
	When $R_0$ is absolutely continuous and $D$ is constant in time, the solution to Problem~\ref{gmod} takes the form
	\begin{equation}
		\bar{R}_t ~=~ \left[ (1-\sigma(t)) \, \mathcal{I} ~+~ \sigma(t) \, \bar{M} \right]_\# R_0
	\end{equation}
	where $\mathcal{I}$ is the identity map on $\R^n$, $\bar{M}$ is the optimal transport map from $R_0$ to $D$, and
	\begin{equation}
		\sigma(t) ~=~ 1 - \cosh \left( (T-t) / \sqrt{\alpha} \right) .
	\end{equation}
	Furthermore, the solution attains the cost
	\begin{equation} \label{opt_cost}
		\cJ(R_0,T;\alpha;D) = \frac{1}{2} \W_2^2(R_0,D) \sqrt{\alpha} \tanh \lp T / \sqrt{\alpha} \rp .
	\end{equation}
\end{lem}

\begin{sketch}
	Follows directly as outlined above.
\end{sketch}

This gives us a portion of the main result stated in Theorem \ref{main_thm}. However, it is not the complete picture -- we still need to recover the feedback form of the optimal controller. To do this, we need to find the velocity field $\bar{V}$ in terms of $\bar{R}_t$ and $D$. Observe that we already have the optimal trajectory
\begin{equation}
	\bar{R}_t ~=~ \left[ (1-\sigma(t)) \, \mathcal{I} ~+~ \sigma(t) \, \bar{M} \right]_\# R_0
\end{equation}
in $\bbW_2$, which corresponds to the trajectory
\begin{equation} \label{opt_traj_o}
	\bar{\phi}_t ~=~ (1-\sigma(t)) \, \mathcal{I} ~+~ \sigma(t) \, \bar{M}
\end{equation}
in $\mathbb{O}$. We find the tangent of this trajectory to be
\begin{equation} \label{tangent_o}
	u_t ~:=~ \phi_t' ~=~ \sigma'(t) (\bar{M} - \mathcal{I}).
\end{equation}
Now, we want to write this tangent in terms of $\bar{M}$ and $\bar{\phi}_t$ to get it in feedback form. Observe that from \eqref{opt_traj_o} we have
\begin{equation}
	\mathcal{I} ~=~ \frac{- \sigma(t)}{1 - \sigma(t)} \bar{M} + \frac{1}{1-\sigma(t)} \bar{\phi}_t ,
\end{equation}
which we substitute into \eqref{tangent_o} to obtain
\begin{equation}
	u_t ~=~ \sigma'(t) \left( \frac{1}{1 - \sigma(t)} \right) (\bar{M} - \bar{\phi}_t).
\end{equation}
Using our coordinate transformation $u_t = v_t \of \bar{\phi}_t$, we have
\begin{equation} \label{vt_feedback_form}
	v_t ~=~ \sigma'(t) \left( \frac{1}{1 - \sigma(t)} \right) \left( \bar{M} \of \bar{\phi}_t^{-1} - \mathcal{I} \right).
\end{equation}
Since $\bar{\phi}$ is a geodesic in $\mathbb{O}$, we know that $\bar{M}_t = \bar{M} \of \bar{\phi}_t^{-1}$ is the optimal transport map taking $\bar{R}_t$ to $D$. Using this fact, we get the optimal controller in error-feedback form:
\begin{equation}
	v_t ~=~ - f(t) \big( \mathcal{I} - \bar{M}_t \big) / \alpha ,
\end{equation}
where
\begin{equation}
	f(t) / \alpha ~=~ \sigma'(t) \left( \frac{1}{1 - \sigma(t)} \right) .
\end{equation}
This completes the development of the results presented in Theorem \ref{main_thm}.

\section{Discussion/Conclusion}

In this paper, we presented a geometric perspective on continuum swarm tracking control and used this perspective to characterize solutions to our model in the $n$-dimensional case where the resource distribution is absolutely continuous and the demand distribution is static.

Where exactly do these results fail under more general assumptions? When the resource distribution is not absolutely continuous (i.e. if it has discrete components), then we can still formulate both the original problem \eqref{orig_eq} and the geometric problem \eqref{gmod_eq}, and we can still solve the geometric problem in the same manner as was done here. However, the failure is in the equivalence of these models: if the resource distribution has discrete components, then the constraint that the trajectory be absolutely continuous is strictly weaker than the dynamic constraint in the original problem. In particular, absolute continuity of the curve allows discrete masses to be broken up. When the demand is time-varying, we can also still formulate the original problem \eqref{orig_eq} and the geometric problem \eqref{gmod_eq}, and here they are still equivalent. However, the failure in this case is in our solution technique: the resource distribution no longer traverses a single geodesic, having different and more complicated necessary conditions for optimality instead. The exploration of both of these cases will be the subject of future work.

We also briefly mentioned how our solution based on a static demand distribution can be used as the basis for a model-predictive control scheme for tracking time-varying, stochastic, and apriori unknown demand trajectories. Developing a numerical control scheme along these lines will also be treated in future work.

Lastly, the robustness of these solutions with respect to noise and approximation error (in particular, that which results from using a continuum model) will be investigated in future work as well.


\appendices

\section{The Scalar LQ Tracking Problem} \label{slqt_prob}

In our previous work \cite{Emerick2022}, we showed that in the special case of one spatial dimension (i.e. distributions on $\R$), Problem~\ref{orig} could be transformed into an infinite-dimensional linear-quadratic (LQ) tracking problem which could then be decoupled into a superposition of scalar LQ tracking problems. The scalar LQ tracking problem plays an important role in this paper, and so we review its solution here. The problem takes the following form.

\begin{problem}[Scalar LQ Tracking Problem] \label{scalar_lq_prob}
	Given an initial state $\eta_0 \in \R$ and tracking signal $\zeta: [0,T] \to \R$, solve
	\begin{equation} \label{lqprob_eq}
		\begin{split}
			& \inf_u\; \int_0^T \big( \zeta(t) - \eta(t) \big)^2 + \alpha u^2(t) \,dt \\
			& \quad \text{s.t.} \quad \dot{\eta}(t) = u(t) ,
		\end{split}
	\end{equation}
	where $\alpha > 0$ is the trade-off parameter and $T$ is the time horizon.
\end{problem}

This problem has been addressed in~\cite{Emerick2022} with  the following solution.
\begin{prop}[\cite{Emerick2022}] \label{scalar_lq_soln}
	The solution to Problem~\ref{scalar_lq_prob} is given implicitly by the non-causal feedback controller 
	\begin{equation}
		\bar{u}(t) ~=~ - f(t) \, \eta(t) / \alpha ~-~ g(t) / \alpha
	\end{equation}
	where
	\begin{align}
		f(t) &~=~ \sqrt{\alpha} \tanh \left( (T-t) / \sqrt{\alpha} \right) \\
		g(t) &~=~ \int_T^t \Phi^{-1}(t,\tau) ~\zeta(\tau) \,d \tau \\
		\Phi(t,0) &~=~ \cosh \left( (T-t) / \sqrt{\alpha} \right) .
	\end{align}
\end{prop}


In the special case where the demand distribution is static, we can sharpen the result as follows.

\begin{cor}[\cite{Emerick2022}] \label{scalar_lq_static_soln}
	When $\zeta(t) = \zeta$ is constant in time, the solution to Problem~\ref{scalar_lq_prob} takes the form
	\begin{equation}
		u(t) ~=~ - f(t) \, \big( \eta(t) - \zeta \big) / \alpha
	\end{equation}
	where
	\begin{equation}
		f(t) ~=~ \sqrt{\alpha} \tanh \left( (T-t) / \sqrt{\alpha} \right) .
	\end{equation}
	The trajectory generated by this control law is
	\begin{equation}
		\eta(t) ~=~ \Phi(t,0) \, \eta_0 ~+~ (1-\Phi(t,0)) \, \zeta
	\end{equation}
	where
	\begin{equation}
		\Phi(t,0) ~=~ \cosh \left( (T-t) / \sqrt{\alpha} \right) ,
	\end{equation}
	and the cost attained with this control law is
	\begin{equation}
		\mathcal{J}' (\eta_0,\zeta;\alpha;T) ~=~ (\zeta - \eta_0)^2 \, \sqrt{\alpha} \, \tanh \big( T / \sqrt{\alpha} \big) \, / \, 2 .
	\end{equation}
\end{cor}





\bibliographystyle{ieeetr}
\bibliography{library}

\end{document}